\documentclass[12pt,a4paper]{amsart}
\usepackage{amsfonts,amssymb,amsmath,amsthm}
\usepackage{url}
\usepackage{enumerate}
\usepackage{bbm}
\usepackage{amssymb}
\usepackage{mathrsfs}
\usepackage{comment}
\usepackage{hyperref}
\usepackage[usenames]{xcolor}
\usepackage[left=3.5cm,right=3.5cm,bottom=4cm]{geometry}

\newtheorem{theorem}{Theorem}[section]
\newtheorem{lemma}[theorem]{Lemma}

\theoremstyle{definition}
\newtheorem{definition}[theorem]{Definition}
\newtheorem{remark}[theorem]{Remark}
\numberwithin{equation}{section}

\DeclareMathOperator{\disp}{disp}

\newcommand{\N}{\ensuremath{{\mathbb N}}}

\newcommand{\F}{\ensuremath{{\mathbb F}}}
\newcommand{\Z}{\ensuremath{{\mathbb Z}}}
\newcommand{\Pro}{\ensuremath{{\mathbb P}}}

\newcommand{\vol}{\mathrm{vol}}

\newcommand{\eps}{\varepsilon}

\author{Mario Ullrich \and Jan Vyb\'iral}
\date{\today}
\address[Mario Ullrich]{Institut f\"ur Analysis\\
Johannes Kepler Universit\"at Linz\\
Altenbergerstr.~69\\
4040 Linz\\
Austria}
\email{mario.ullrich@jku.at}
\address[Jan Vyb\'iral]{Department of Mathematics\\
Czech Technical University\\
Trojanova 13\\
12000 Praha \\
Czech Republic}
\email{jan.vybiral@fjfi.cvut.cz}
\thanks{We would like to express our gratitude to the Erwin Schr\"odinger International Institute for
Mathematics and Physics for its hospitality during the programme on ``Tractability of High Dimensional
Problems and Discrepancy'', where some part of this research was carried out.
We also gratefully acknowledge the support of the Oberwolfach Research Institute for Mathematics,
where initial discussion were held during the workshop ``Perspectives in High-Dimensional Probability and Convexity''.
We thank also Michael Gnewuch, Daniel Kr\'al' and Hemant Tyagi for fruitful discussions.
JV was supported by the grant P201/18/00580S
of the Grant Agency of the Czech Republic and by
the European Regional Development Fund-Project ``Center for Advanced Applied Science'' (No. CZ.02.1.01/0.0/0.0/16\_019/0000778)}


\begin{document}

\title[Deterministic sets with small dispersion]{Deterministic constructions of high-dimensional sets with small dispersion}

\begin{abstract}
The dispersion of a point set $P\subset[0,1]^d$ is the volume of the largest box with 
sides parallel to the coordinate axes, which does not intersect $P$. 
Here, we show a construction of low-dispersion point sets, which can be deduced
from solutions of certain $k$-restriction problems, 
which are well-known in coding theory.

It was observed only recently that, for any $\varepsilon>0$,
certain randomized constructions provide point sets
with dispersion smaller than $\varepsilon$ and number of elements growing only 
logarithmically in $d$.
Based on deep results from coding theory, we present explicit, deterministic 
algorithms to construct such point sets 
in time that is only polynomial in~$d$.
Note that, however, the running-time will be super-exponential in $\eps^{-1}$.

\end{abstract}
\maketitle



\section{Introduction and Results}

For $d\in\N$ and a point set $P\subset[0,1]^d$ we define  
the \emph{dispersion of $P$} by
\[
\disp(P) \;:=\; \sup_{B\colon B\cap P=\varnothing}\, |B|, 
\]
where the supremum is over all axis-parallel boxes $B=I_1\times\dots\times I_d$ 
with intervals $I_\ell\subset[0,1]$, 
and $|B|$ denotes the (Lebesgue) volume of $B$. 
As we are interested in point sets which make the above quantity as small as 
possible, we additionally define, for $n,d\in\N$, 
the \emph{$n$-th-minimal dispersion} 
\[
\disp(n,d) \;:=\; \inf_{\substack{P\subset[0,1]^d\colon\\ \#P=n}}\, \disp(P)
\]
and, for $\eps\in(0,1)$, its inverse function
\[
N(\eps,d) \;:=\; \min\Bigl\{n\colon \disp(n,d)\le\eps\Bigr\}.
\]
Hence, $N(\eps,d)$ is the minimal cardinality of a point set $P\subset[0,1]^d$ 
that has dispersion smaller than $\eps$.

Besides the fact that the above geometric quantities are interesting in its 
own right, they also attracted attention in the numerical analysis community, 
especially when it comes to very high dimensional applications. 
The reason is that bounds on the (minimal) dispersion lead to bounds on 
worst-case errors, and hence the complexity, for some numerical problems,
including optimization and approximation in various settings, 
see~\cite{BDDG14,Niederreiter83,NR15,RT96,Te17b,Te17c,YLV00}. 
This is a similar situation as for the much more studied \emph{discrepancy}, 
which corresponds to certain numerical integration problems, 
see e.g.~\cite{DP10,DP14,DT97,Niederreiter92,No15,NW10}.

Moreover, bounding the dispersion is clearly also related to the problem 
of \emph{finding} the largest empty box. In dimension two, this is the 
\emph{Maximum Empty Rectangle Problem}, which is one of the oldest problems 
in computational geometry. For the state of the art and further references 
we refer to~\cite{DJ13,DJ13b,DJ16,DJ18,NLH84}.

Regarding bounds on the inverse of the minimal dispersion, 
we have
\[
\frac{\log_2(d)}{8\eps} \;\le\; N(\eps,d) \;\le\; \frac{C^d}{\eps}
\]
for some $C<\infty$ and all $\eps<1/8$, see~\cite{AHR15}, where the upper bound is attained 
by certain digital nets. See also~\cite{DJ13b,RT96} for a related 
upper bound.
Although these bounds show the correct dependence on $\eps$, 
they are rather bad with respect to the dimension $d$. 
This gap was narrowed in the past years by several authors, 
see~~\cite{Kr17,Ru17,Sosno,UV}, 
including the important work of Sosnovec who proved that the logarithmic 
dependence on $d$ is optimal. 
With this respect, the best bound at present is
\begin{equation}\label{eq:bound}
\frac{\log_2(d)}{8\eps} \;\le\; N(\eps,d) 
\;\le\; 2^7\,\log_2(d)\,\frac{\bigl(1+\log_2(\varepsilon^{-1})\bigr)^2}{\varepsilon^{2}},
\end{equation}
see~\cite{UV} and Theorem~\ref{thm:UV18} below.
Note that the logarithmic dependence is special for the cube, as it is known
that, for the same problem on the torus, we have a  lower bound linear in $d$, 
see~\cite{MU18}.

The main drawback of these results is that they only show the existence of 
point sets with small dispersion. 
The only explicit constructions we are aware of are the above mentioned 
digital nets, which lead to a bad $d$-dependence, 
and sparse grids, which satisfy the upper bound 
$N(\eps,d)\le(2d)^{\log_2(1/\eps)}$, see~\cite{Kr17}. 
It is already clear from the number of points, 
that both do not lead to constructions of point sets with small dispersion 
that can be carried out in time that is polynomial in $d$.
Moreover, as the existence proofs are based on random points on a finite grid, 
one could use the \emph{naive algorithm}: 
try each of the possible configurations, calculate 
its dispersion (if possible) and output a set that satisfies 
the requested bound. 
The running-time of this algorithm is (in the worst case) clearly at least
exponential in $d$.

\begin{remark}
Note that the decision problem, 
if a given point set has \emph{discrepancy} smaller than $\eps$, 
is known to be NP-hard, see~\cite{GSW09} or \cite[Section 3.3]{DGW14}, and the same is 
true for the dispersion, if the dimension $d$ is part of the input, cf. \cite{BK}.

For upper bounds on the cost of constructing points with small discrepancy 
and further literature, see~\cite{DGW10, DGW14, G10}. However, all known algorithms 
for this problem so far have running-time at least exponential in $d$. 
We hope that the results of this paper will lead to some progress also 
for this problem.
\end{remark}
\smallskip

Reconsidering the bound~\eqref{eq:bound}, one may hope that points 
with small dispersion may be constructable also in very high dimensions. 
Ideally, we would like to have algorithms for the construction of point 
sets of size $N=\mathcal{O}(N(\eps,d))$ with dispersion at most $\eps>0$, 
whose computational cost is polynomial in $\eps^{-1}$ and $d$. 
However, this seems to be out of reach. 
(Note that the output already costs $d\cdot N$.)

Here, we focus on the dependence on $d$ and show, 
using deep results from the theory of error-correcting codes, 
that
point sets with small dispersion of size $\sim\log(d)$ can be constructed 
in time that is polynomial in $d$.
Unfortunately, we do not have a good control on the dependence on $\eps^{-1}$. 
It remains an open problem to find \emph{fully-polynomial constructions} 
for the dispersion.

Our two main results are derandomized versions of the results from~\cite{Sosno} 
and~\cite{UV}. Both use different approaches and lead to somewhat different 
results. The first one, as discussed in Section~\ref{sec:sos}, 
leads to point sets of size $\mathcal{O}_\eps(\log d)$ that can be 
constructed in time linear in the size of the output, i.e., in time 
$\mathcal{O}_\eps(d\log d)$, see Algorithm~1 and Theorem~\ref{thm:main1}.
Here, and in the following, $\mathcal{O}_\eps(\log d)$ means that 
the implied constant depends on $\eps$ in an unspecified way.
A second, and much more involved, construction will be given by Algorithm~2 
in Section~\ref{sec:UV}. The corresponding result reads as follows.

\medskip
\noindent
{\bf Theorem 4.4}
{\it
Let $\varepsilon\in(0,\frac{1}{4}]$ and $d\ge 2$. 
Then, there is an absolute constant $C<\infty$, such that
Algorithm 2 constructs a set $P\subset[0,1]^d$ with $\disp(P)\le \varepsilon$ and 
\[
\#P\le C \left(\frac{1+\log_2(\varepsilon^{-1})}{\varepsilon}\right)^6 \log(d^*),
\]
where $d^*:=\max(d, 2/\eps)$. 

The running-time of Algorithm~2 is $\mathcal{O}_\eps(d^c)$ for some $c<\infty$.
}
\medskip

Note that, in contrast to Algorithm~1, the point set that is constructed by 
Algorithm~2 has size that is polynomial in $\eps^{-1}$. However, as the proof 
shows, its computational cost is much larger.

Our general approach is as following. We start with a detailed inspection of the random constructions
of point sets with small dispersion. This will allow us to clearly separate the setting of the construction
and its randomized part. It will turn out then, which properties are crucial for each of the approaches
and what exactly is the role of randomness. Afterwards, we replace the randomized part by a deterministic
one.


\section{Basics from Coding Theory}

Our deterministic constructions of point sets with small dispersion
will be essentially obtained by certain ``derandomization'' of
recent proofs from \cite{Sosno,UV}. As we will rely on rather deep tools from coding theory,
we summarize in this section the necessary definitions and results
for later use.

\subsection{Universal sets}\label{subsec:universal}

First, we introduce the concept of \emph{$(n,k)$-universal sets},
which is also known in coding theory under the name of \emph{$t$-independent set problem}.
It has its roots in testing of logical circuits \cite{SB}.
\begin{definition}[$(n,k)$-universal sets]\label{def:nkuniversal}
Let $1\le k\le n$ be positive integers. We say that $T\subset\{0,1\}^n$ is an $(n,k)$-universal set,
if for every index set $S\subset \{1,\dots,n\}$ with $\#S=k$, the projection of $T$ on $S$ contains
all possible $2^k$ configurations.
\end{definition}
Naturally one is interested in (randomized and deterministic) constructions
of small $(n,k)$-universal sets. The straightforward randomized construction
provides the existence of an $(n,k)$-universal set of size $\lceil k2^k\log(n)\rceil$.
On the other hand, \cite{KS} gives a lower bound on the size of an $(n,k)$-universal set 
of the order $\Omega(2^k\log(n))$.
There exist several deterministic constructions of $(n,k)$-universal sets in the literature (cf. \cite{Alon, Alonall, NN}) and we shall
rely on the results given in \cite{NSS}.

\begin{theorem}[{\cite[Theorem 6]{NSS}}]\label{thm:NSS:cite}
There is a deterministic construction of an $(n,k)$-universal set of size 
$2^{k+O(\log^2(k))} \log(n)$,
which can be listed in linear time of the length of the output.
\end{theorem}

Although the notion of an $(n,k)$-universal set is not very flexible and comes from a different
area of mathematics, we will see in Section~\ref{sec:sos}, that there is indeed a link to sets with small dispersion.
Reusing the known results from coding theory,
it will already allow us to obtain our first deterministic construction of a point set with cardinality
of order $\log(d)$. However, in this approach we have only very limited control of the dependence on $\eps$.

We also need the following natural generalization of $(n,k)$-universal sets. 
We could not find this concept in the literature, but we assume that this 
and the proceeding lemma are known.

\begin{definition}[$(n,k,b)$-universal sets]\label{def:nkbuniversal}
Let $1\le k\le n$ and $b\ge 2$ be positive integers. 
We say that $T\subset\{0,1,\dots,b-1\}^n$ is an $(n,k,b)$-universal set,
if for every index set $S\subset \{1,\dots,n\}$ with $\#S=k$, 
the projection of $T$ on $S$ contains all possible $b^k$ configurations.
\end{definition}

If $b=2$, Definitions~\ref{def:nkuniversal} and \ref{def:nkbuniversal} coincide, 
i.e. $(n,k,2)$-universal sets are just the usual $(n,k)$-universal sets.
We use the following two observations to transfer the known results about 
$(n,k)$-universal sets to our setting.

\begin{lemma}\label{lem:uset}
Let $1\le k\le n$ and $b\ge 2$ be positive integers.
\begin{enumerate}
\item[$(i)$] Let $T\subset \{0,1,\dots,b\}^n$ be an $(n,k,b+1)$-universal set. 
	Then there is an $(n,k,b)$-universal set $T'\subset\{0,1,\dots,b-1\}^n$ 
	of at most the same size.
\item[$(ii)$] Let $m\in\N$ and $T\subset\{0,1\}^{m n}$ be an $(m n,m k)$-universal set. 
	Then there is an $(n,k,2^m)$-universal set of the same size.
\end{enumerate}
\end{lemma}
\begin{proof}
The proof is quite straightforward. To show $(i)$, just replace all occurrences of $b$ 
among the coordinates of $T$ by zero.
For the proof of the second part it is enough to interpret each 
$x\in T\subset\{0,1\}^{m n}$ as a digital representation
of $\tilde x\in\{0,1,\dots,2^{m}-1\}^n$.
\end{proof}

The direct random construction yields the existence of an $(n,k,b)$-universal set
of the size $\lceil kb^k\log(ebn/k)\rceil$. Using Theorem \ref{thm:NSS:cite},
we can easily obtain a deterministic construction of an $(n,k,b)$-universal set of only a slightly larger size.

\begin{theorem}\label{thm:NSS}
There is a deterministic construction of an $(n,k,2^m-1)$-universal set of size 
$2^{mk+O(\log^2(mk))} \log(n)$,
which can be listed in linear time of the length of the output.
\end{theorem}

\begin{proof}
By Theorem \ref{thm:NSS:cite}, there is a construction of an 
$(mn,mk)$-universal set with at most $2^{mk+O(\log^2(mk))}\log(mn)$ elements.
The result then follows from Lemma~\ref{lem:uset}.
\end{proof}

\subsection{$k$-restriction problems}\label{subsec:restriction}

For the derandomization of the analysis of \cite{UV}
we need a more flexible notion of
the so-called \emph{$k$-restriction problems}, 
see~\cite[Section~2.2]{NSS}.
Solutions to these problems will be one of the building blocks of 
our deterministic construction of sets with small dispersion 
whose size is polynomial in $1/\eps$ and, still, logarithmic in~$d$.

\begin{definition}[$k$-restriction problems]\label{def:krestriction} 
Let $b,k,n,M$ be positive integers and let 
${\mathcal C}=\{C_1,\dots,C_M:C_i\subset\{0,1,\dots,b-1\}^k\}$
be invariant under the permutations of the index set $\{1,\dots,k\}$. 
We say that $T=\{x^1,\dots,x^N\}\subset \{0,1,\dots,b-1\}^n$ 
satisfies the $k$-restriction problem with respect to ${\mathcal C}$,
if for every $S\subset\{1,\dots,n\}$ with $\#S=k$ and for every 
$j\in\{1,\dots,M\}$, there exists $x^\ell\in T$ with $x^\ell|_S\in C_j.$
\end{definition}

Definitions~\ref{def:nkuniversal} and \ref{def:nkbuniversal} are indeed 
special cases of Definition~\ref{def:krestriction}. 
To show this, let us choose $M=b^k$ and let $C_1,\dots,C_M$ be all the 
different singleton subsets of $\{0,\dots,b-1\}^k$.
Then, $T$ satisfies the $k$-restriction problem with respect to ${\mathcal C}$ 
if for every index set $S\subset \{1,\dots,n\}$ and every possible 
$v\in\{0,\dots,b-1\}^k$ there exists an $x\in T$ with $x|_S=v$, i.e.~if the 
restriction of $T$ to every index set $S$ with $k$ elements
attains all $b^k$ possible values. 

An important parameter of the $k$-restriction problems is the minimal size of each of the restriction sets $C_j$, i.e.
\begin{equation}\label{eq:c}
c=c(\mathcal{C}):=\min_{1\le j\le M}\#C_j.
\end{equation}

Random constructions of sets satisfying the $k$-restriction problem with 
parameters $(b,k,n,M)$ and ${\mathcal C}=\{C_1,\dots,C_M\}$
are based on a simple union bound. 
Indeed, let $1\le j\le M$ and $S\subset \{1,\dots,n\}$ with $\#S=k$ be fixed. 
The probability, that a randomly chosen vector $v\in\{0,1,\dots,b-1\}^n$ 
satisfies $C_j$ on $S$ is at least
\[
{\mathbb P}(v\text{ satisfies }C_j\text{ at }S)
\,=\, \frac{\#C_j}{b^k}\ge \frac{c}{b^k}.
\]
If we choose $N$ random vectors $v_1,\dots,v_N$ independently, the probability that none of them satisfies $C_j$ at $S$ is at most
\[
{\mathbb P}(\text{no }v_1,\dots,v_N\text{ satisfies }C_j\text{ at }S)
\,=\, \Bigl(1-\frac{\#C_j}{b^k}\Bigr)^N\le \Bigl(1-\frac{c}{b^k}\Bigr)^N.
\]
Finally, the probability that there is a set $S\subset\{1,\dots,n\}$ with $\#S=k$ and $1\le j\le M$, such that no $v_1,\dots,v_N$
satisfies $C_j$ at $S$ is, by the union bound, at most
\[
\binom{n}{k}\cdot M\cdot \Bigl(1-\frac{c}{b^k}\Bigr)^N.
\]
This expression is smaller than one if
$$
N \,>\, \frac{\log(n^k M)}{\log\bigl(\frac{b^k}{b^k-c}\bigr)}.
$$
This means that there exist solutions to a $k$-restriction problem 
with parameters $(b,k,n,M)$ of size $N$ whenever
\[
N \,\ge\, \left\lceil\frac{b^k}{c} \log(n^k M)\right\rceil, 
\]
where $c$ is from \eqref{eq:c}.
Theorem~1 of \cite{NSS} states that there is a deterministic algorithm 
that outputs such a solution of size equaling the union bound.
The main idea of its proof is that the random sampling can be replaced
by an extensive search through a $k$-wise independent probability space
with $n$ random variables with values in $\{1,\dots,b\}.$
\begin{theorem}[{\cite[Theorem 1]{NSS}}]\label{thm:NSS2}
For any $k$-restriction problem with parameters $(b,k,n,M)$ with $b\le n$, 
there is a deterministic algorithm that outputs a collection 
obeying the $k$-restrictions, with the size of the collection 
equaling 
\[
\left\lceil\frac{b^k}{c} \log(n^k M)\right\rceil,
\]
where $c$ is from \eqref{eq:c}.
The time taken to output the collection is
\[
\mathcal{O}\left(\frac{b^k}{c} \left(\frac{e n^2}{k}\right)^k M\, \mathcal{T}\right), 
\]
where $\mathcal{T}$ is the time complexity of the membership oracle. 
\end{theorem}
Here, the membership oracle is a procedure which, for given $v\in\{1,\dots,b\}^n$, $S\subset\{1,\dots,n\}$
with $\#S=k$ and $j\in\{1,\dots,M\}$, outputs if the restriction of $v$ on $S$ belongs to $C_j.$
In what follows it can be executed in ${\mathcal{T}}=O(k)$ time.

\subsection{Splitters}\label{subsec:splitter}

The last ingredient of our derandomization procedure are splitters.
They played a central role in \cite{NSS} as the basic building blocks of all deterministic constructions
given there. Essentially, they allow to split a large problem into smaller problems which can then be
treated by the extensive search of Theorem \ref{thm:NSS2}.

\begin{definition}[$(n,k,l)$-splitter]\label{def:splitters} 
Let $n,k,l$ be positive integers. An $(n,k,l)$-splitter $H$ is a family of functions from $\{1,\dots,n\}$
to $\{1,\dots,l\}$, such that for every $S\subset\{1,\dots,n\}$ with $\#S=k$ there is $h\in H$, which splits $S$ perfectly.
It means that the sets $h^{-1}(\{j\})\cap S$ are of the same size for all $j\in\{1,\dots,l\}$ (or as similar as possible if $l\nmid k$).
\end{definition}

Similarly to \cite{NSS} and \cite{AMS}, we will rely on $(n,k,k^2)$-splitters.
By Definition \ref{def:splitters}, $A(n,k)$ is an $(n,k,k^2)$-splitter, 
if it is a collection of mappings $a:\{1,\dots,n\}\to\{1,\dots,k^2\}$ 
such that for every $S\subset \{1,\dots,n\}$ with $\#S=k$, 
there is an $a\in A(n,k)$, which is injective on $S$. 

Small $(n,k,k^2)$-splitters can be obtained from asymptotically good error 
correcting codes, see~\cite[Lemma 3]{AMS}.
Indeed, let $c_1,\dots,c_n\in\{1,\dots,k^2\}^L$ denote the codewords of an error correcting code of length $n$
over the alphabet $\{1,\dots,k^2\}$ and 
normalized Hamming distance at least $1-\frac{2}{k^2}$.
This is, $c_i$ and $c_j$ ($i\neq j$) can be equal on at most $2L/k^2$ coordinates.
If now $S\subset\{1,\dots,m\}$ with $\#S=k$, then there are $k(k-1)/2$ 
pairs $(i,j)$ with $i\neq j\in S$.
As $\frac{2L}{k^2}\cdot \frac{k(k-1)}{2}<L$, there must be a coordinate,
where all the codewords differ from each other. Finally, if we consider the mappings
$h_j:i\to c_i(j)$, $i=1,\dots,n$ we observe that 
$H=\{h_j:j=1,\dots,L\}$ is an $(n,k,k^2)$-splitter of size $L$.

Such explicit codes exist by \cite{Alonall} with $L=\mathcal{O}(k^4\log n)$. 
To see this, note that the rate $R$ of a code as above is defined by 
$R:=\log_{k^2}(n)/L=\log(n)/(L\log(k^2))$. 
By \cite[eq.~(5)]{Alonall}, see also~\cite[Lemma 3]{AMS}, 
we obtain that an explicit code with 
normalized Hamming distance at least $\delta=1-\frac{2}{k^2}$ 
exist with 
\[
R \,\ge\, \gamma_0 \left(1-H_{k^2}\left(1-\frac{1,5}{k^2}\right)\right) \left(1-\frac{1-2/k^2}{1-1,5/k^2}\right), 
\]
where $\gamma_0>0$ is an absolute constant and 
$H_q(x):=-x\log_q(x)-(1-x)\log_q(1-x)+x\log_q(q-1)$.
Simple computations show that this implies $R\ge c/(k^4 \log(k^2))$ for some $c>0$. 
This, in turn, implies that we can choose $L=\mathcal{O}(k^4\log n)$.

The explicit construction of~\cite{Alonall} 
yields a linear code that is
based on a two-fold 
concatenation code that combines the Wozencraft ensemble, Justesen codes 
and expander codes, which in turn rely on famous deterministic constructions 
of expander graphs~\cite{LPS}.
This construction is 'uniformly constructive' (see~\cite{Alonall}), 
i.e., the construction 
can be done in time growing only polynomially in $n$. 
Furthermore, the code satisfies~\cite[eq.~(5)]{Alonall} for all $\delta<1-1/q$, where $q=k^2$ in our case.
Note also that the running-time depends only polynomially also on $k$, cf. \cite[Thm.~3(iv)]{NSS}.
For the details we refer to Sections 3 and 4 of~\cite{Alonall}.

The following lemma summarizes the discussion above and shows that 
$(n,k,k^2)$-splitters of relatively small size can be constructed 
explicitly in polynomial time.

\begin{lemma}[{cf.~\cite[Lemma~3]{AMS}}]\label{lem:splitter}
There is an explicit $(n,k,k^2)$-splitter 
of size 
\[
\mathcal{O}\bigl(k^4\log(n)\bigr)
\]
that can be constructed in polynomial time in $n$ and $k$.
\end{lemma}

\medskip

The $(n,k,k^2)$-splitters can be used whenever $n$ is ``very large'' compared 
to $k$. Roughly speaking, and in the context of the present paper,
we will transform a $k$-restriction problem of size $n$ to a 
$k$-restriction problem of size $k^2$, which can then be solved 
using the results of Section~\ref{subsec:restriction}. 
This will lead to construction algorithms with an apparently optimal 
dependence of their running time on the original problem size $n$.
This approach was already used to prove~\cite[Theorem~6]{NSS}, 
see Theorem~\ref{thm:NSS}.

\medskip


\section{Derandomization of Sosnovec's proof}\label{sec:sos}

First, we consider the construction of Sosnovec~\cite{Sosno}, 
which gives logarithmic dependence of $N(\varepsilon,d)$ on $d$
but involves no special control of its dependence on $\varepsilon^{-1}$. 
His main theorem was the following.

\begin{theorem}(\cite[Theorem 2]{Sosno})\label{Theo:Sosnovec} For every $\varepsilon\in(0,\frac{1}{4}]$, there exists a constant $c_\varepsilon>0$, such that
for every $d\ge 2$ there is a point set $P\subset[0,1]^d$ with $\disp(P)\le\varepsilon$ and 
\[
\#P \,\le\, c_\varepsilon\log(d).
\]
\end{theorem}
\medskip

We will see that this result can be essentially derandomized using results from coding theory, 
while loosing only a negligible factor.
The drawback of this approach is the extremely bad dependence of $c_\varepsilon$ on $\varepsilon$.
We sketch the main ideas of Sosnovec's proof.

\subsection{Sosnovec's proof - the setting}
For an integer $m\ge 2$ with $2^{-m}\le \varepsilon<2^{-m+1}$, we define
$$
{\mathbb M}_m=\Bigl\{\frac{1}{2^m},\dots,\frac{2^m-1}{2^m}\Bigr\}.
$$
The point set constructed will be a subset of ${\mathbb M}_m^d$. Furthermore, we define
\begin{equation}\label{eq:Omegam}
\Omega_m:=\Big\{\mathbb B=I_1\times\dots\times I_d\subset[0,1]^d\,:\,\vol(\mathbb B)>\frac{1}{2^m}\Big\}
\end{equation}
to be the set of all boxes with sides parallel to the coordinate axes and volume larger than $2^{-m}.$

Let ${\mathbb B}=I_1\times\dots\times I_d\in\Omega_m$. The key observation of \cite{Sosno} is that the number of indices $j\in\{1,\dots,d\}$ with ${\mathbb M}_m\not\subset I_j$
is bounded from above by $m2^m$, a quantity independent on $d$. To be more specific, if we denote
$$
{\mathcal A}({\mathbb B})=\{j\in\{1,\dots,d\}:{\mathbb M}_m\not\subset I_j\},
$$
then $\#{\mathcal A}({\mathbb B})\le A_m:=\min(m2^m,d)$ for every $\mathbb{B}\in\Omega_m$. 
We will refer to ${\mathcal A}({\mathbb B})$ as the set of ``active indices'' of ${\mathbb B}.$
If ${\mathcal A}({\mathbb B})$ is not of the full possible size, we enlarge it by adding any of the other indices to obtain a set with cardinality equal to $A_m$.
Therefore, we can associate to each ${\mathbb B}\in\Omega_m$ (possibly in a non-unique way) a set ${\mathcal A}$ with $\#{\mathcal A}=A_m$ and a vector
$z\in {\mathbb M}_m^{A_m}$ such that any $x\in{\mathbb M}_m^d$ with $x|_{\mathcal A}=z$ lies in ${\mathbb B}.$

Vice versa, if we have a point set $P=\{x^1,\dots,x^N\}\subset {\mathbb M}_m^d$, 
such that for every ${\mathcal A}\subset \{1,\dots,d\}$
with $\#{\mathcal A}=A_m$ and to every $z\in {\mathbb M}_m^{A_m}$, 
there is some $x^j\in P$ with $x^j|_{\mathcal A}=z$,
then, by what we just said, $P$ intersects every ${\mathbb B}\in\Omega_m$.
Therefore, the dispersion of $P$ can not be larger than $2^{-m}$, 
i.e., $\disp(P)\le 2^{-m}$ and hence $N(2^{-m},d)\le N$.

To simplify the combinatorial part later on,
we multiply all coordinates by $2^m$, which results to vectors with integer components.
This motivates the following definition.

\begin{definition}\label{defn:condS} Let $m\ge 2$. We say that $T=\{x^1,\dots,x^N\}\subset\{1,2,\dots,2^{m}-1\}^d$
satisfies \emph{the condition (S) of the order $m$} if for every 
${\mathcal A}\subset\{1,\dots,d\}$ with $\#{\mathcal A}=A_m:=\min(m2^m,d)$, 
the set of restrictions $x^1|_{\mathcal A},\dots,x^N|_{\mathcal A}$ contains all 
$(2^m-1)^{A_m}$ possible values.
\end{definition}

By what we said above, anytime $m\ge 2$ and $T=\{x^1,\dots,x^N\}$ satisfies the 
condition (S) of the order $m$, then $P\subset{\mathbb M}_m^d$ with $P=2^{-m}\cdot T$ 
satisfies $\disp(P)\le 2^{-m}$.
The proof of Theorem \ref{Theo:Sosnovec} is therefore finished, once we find a set $T$
with $\#T\le c_m\log(d)$ satisfying the condition (S) of order $m$.

\subsection{Sosnovec's proof - randomized construction}

The rest of the proof in \cite{Sosno} can now be understood as a randomized construction of a small set satisfying the condition (S).
Indeed, there are $\binom{d}{A_m}$ subsets $S$ of $\{1,\dots,d\}$ with $A_m$ elements.
We now fix one such set $S$ and one vector $z\in\{1,\dots,2^m-1\}^{A_m}$. The probability
that a point $x\in\{1,\dots,2^m-1\}^d$ chosen at random (from the uniform distribution) fulfills $x|_{S}=z$
is $(2^m-1)^{-A_m}$. Therefore, the probability that none of the $N$ randomly chosen points $x^1,\dots,x^N$ fulfills this restriction is $[1-(2^m-1)^{-A_m}]^N$.
Finally, the probability that there is a set $S\subset\{1,\dots,d\}$ with $\#S=A_m$ and a vector $z\in\{1,\dots,2^m-1\}^{A_m}$
such that no $\{x^1,\dots,x^N\}$ satisfies $x^j|_{S}=z$ is, by the union bound, at most
\[
\binom{d}{A_m}(2^m-1)^{A_m}[1-(2^m-1)^{-A_m}]^N.
\]
By simple calculus, if $N$ is of the order $c_m\log_2(d)$ for $c_m$ large enough, this expression is smaller than one.
Hence, with positive probability, the randomly chosen point set $\{x^1,\dots,x^N\}$ satisfies the condition (S).

\subsection{Derandomization using universal sets}

Definition \ref{defn:condS} resembles very much the concept of $(n,k)$-universal sets, 
see Section~\ref{subsec:universal}.
In particular, it is easy to see that every $(d,A_m,2^m-1)$-universal set 
-after adding 1 to each coordinate- 
satisfies condition (S) of the order $m$. 

Therefore, we can use Theorem \ref{thm:NSS} to replace the random arguments 
of the last section by a deterministic algorithm. 
We glue all the components together in the form of an algorithm.
\vskip.5cm
\fbox{\begin{minipage}{13cm}
{\bf Algorithm 1}
\begin{enumerate}
\item For $\varepsilon\in(0,\frac{1}{4}]$ and $d\ge2$, choose a positive integer $m$ \\
			with $2^{-m}\le\varepsilon<2^{-m+1}$ and set $A_m:=\min(m2^m,d)$;
\item Generate an $(md, m A_m)$-universal set as in \cite[Theorem~6]{NSS};
\item Interpret these vectors as digital decompositions to obtain \\ 
			an $(d,A_m,2^m)$-universal set;
\item Replace $2^{m}-1$ by $0$ in	all coordinates to obtain \\
			an $(d,A_m,2^m-1)$-universal set;
\item Increase all the coordinates by one; (This set satisfies (S).)
\item Finally, divide all the coordinates by $2^m$ and output the point set.
\end{enumerate}
\end{minipage}}
\vskip.5cm

It remains to consider the running-time of this algorithm.

\begin{theorem}\label{thm:main1}
Let $\varepsilon\in(0,\frac{1}{4}]$ and $d\ge 2$. Then there is a positive constant $c_\varepsilon>0$, such that
Algorithm 1 constructs 
a set $P\subset[0,1]^d$ with $\disp(P)\le \varepsilon$ and 
\[
\#P\,\le\, 
c_{\varepsilon}\log_2(d).
\]
The running time of Algorithm~1 is linear in the length of the output.
\end{theorem}
\medskip

\begin{proof}
The first part of the theorem is proven by what we said above. 
Concerning the running-time, we obtain from Theorem~\ref{thm:NSS}, 
that the $(d,A_m,2^m-1)$-universal set, which we generate in steps (2)--(4) of Algorithm~1, 
can be constructed in linear time and with size
\[
N= 
2^{mA_m+O(\log^2(m A_m))}\log(d)\le 2^{m^2 2^m+O(m^2)}\log_2(d).
\]
The remaining operations can be done in a linear time, without enlarging the point set. \\
\end{proof}
As expected, the dependence of the size of $P$ on $2^m\approx \varepsilon^{-1}$ is rather bad (as it was in \cite{Sosno}), but there is indeed only a logarithmic
dependence on $d$.

\bigskip

\section{Improving the dependence in $\varepsilon^{-1}$}\label{sec:UV}

The main aim of \cite{UV} was to refine the analysis of \cite{Sosno} and to achieve a better dependence of $N(\varepsilon,d)$
on $\varepsilon$, without sacrificing the logarithmic dependence on~$d$.
The main theorem of~\cite{UV} was the following.

\begin{theorem}\label{thm:UV18}
Let $d\ge2$ be a natural number and let $\eps\in(0,1/2)$. 
Then there exists a point set $P\subset[0,1]^d$ with $\disp(P)\le\eps$ and
\[
\#P \,\le\, 2^7\, \log_2(d)\, \frac{(1+\log_2(\eps^{-1}))^2}{\eps^2}.
\]
\end{theorem}
\medskip

Also this result can be derandomized using results from coding theory. 
By doing this, we will lose some power of $1/\eps$ in the size of the 
point set. However, it will still be of order $\log(d)$.

\subsection{Enhanced analysis of the random construction} \label{subsec:UV-random}

The main novelty of \cite{UV} was a more careful splitting of $\Omega_m$ 
(see~\eqref{eq:Omegam}) into subgroups.
To be more specific (and using the notation of \cite{HPUV}), for $s=(s_1,\dots,s_d)\in\{1,\dots,2^{m}-1\}^d$ and 
$p=(p_1,\dots,p_d)\in\mathbb{M}_m^d$, we denoted 
$\Omega_m(s,p)$ to be those cubes from $\Omega_m$, which have $I_\ell$ approximately of the length $\frac{s_\ell}{2^m}$
and its left point around $p_\ell$ for all $\ell=1,\dots,d$, i.e.
\begin{align*}
\Omega_m(s,p)&:=\bigg\{\mathbb B=I_1\times\dots\times I_d\in\Omega_m \,\colon\, 
\forall \ell\in\{1,\dots,d\}: \frac{s_\ell}{2^{m}}<\vol(I_\ell) \le \frac{s_\ell+1}{2^{m}}\\
&\qquad\qquad \text{and}\quad\inf I_\ell\in \Big[p_\ell-\frac{1}{2^{m}},p_\ell\Big)\bigg\}.
\end{align*}

We denote by ${\mathbb I}_m$ the pairs $(s,p)$, for which $\Omega_m(s,p)$ is non-empty. 
It is easy to see that their number is bounded from above by
\begin{equation}\label{eq:card}
\#{\mathbb I}_m\le \exp\Bigl(m2^{m}\log(2^{m+3}d)\Bigr), 
\end{equation}
see~\cite[eq.~(2.1)]{HPUV}.
If $(s,p)\in{\mathbb I}_m$,  we further set
\begin{equation}\label{eq:def:B}
\mathbb B_m(s,p)  \,:=\, \bigcap_{\mathbb B\in \Omega_m(s,p)} \mathbb B
\,=\, \prod_{\ell=1}^d \Bigl[p_\ell,p_\ell+\frac{s_\ell-1}{2^{m}}\Bigr].
\end{equation}

The advantage of dividing of $\Omega_m$ into the groups $\Omega_m(s,p)$ is the surprisingly good control
of the probability that a randomly chosen point $z\in {\mathbb M}_m^d$ lies in the intersection
of all the cubes from $\Omega_m(s,p).$ It is actually, up to a constant, of the same order
as the volume of each of the cubes in $\Omega_m$, i.e. of $2^{-m}$.

\begin{lemma}(\cite[Lemma 3]{UV} and \cite[Lemma 2.1]{HPUV})\label{lem:discrete} 
Let $m\in\N$, $s\in\{0,1,\dots,2^{m}-1\}^d$ and 
$p\in\{1/2^{m},\dots,1-1/2^{m}\}^d$ be such that $(s,p)\in{\mathbb I}_m$. Let 
$z$ be uniformly distributed in $\mathbb M_m^d$. Then
\begin{equation*}
\Pro\big(z\in \mathbb B_m(s,p)\big)\ge \frac{1}{2 ^{m+4}}\,.
\end{equation*}
\end{lemma}
The aim of \cite{UV} was to combine \eqref{eq:card} with Lemma~\ref{lem:discrete} and the union bound.
Indeed, the probability that a randomly chosen point $z\in {\mathbb M}_m^d$ avoids ${\mathbb B}_m(s,p)$ is at most $1-2^{-m-4}.$
Therefore, the probability that a set $P=\{x^1,\dots,x^N\}\subset{\mathbb M}_m^d$ of $N$ randomly and independently generated points
does not intersect ${\mathbb B}_{m}(s,p)$
is at most $(1-2^{-m-4})^N$, i.e.
\[
{\mathbb P}\bigl(\forall \ell\in \{1,\dots,N\}\colon x^\ell\not\in {\mathbb B}_m(s,p)\bigr)
\,\le\, (1-2^{-m-4})^N.
\]
By the union bound over all $(s,p)\in{\mathbb I}_m$, we get further
\[
{\mathbb P}\bigl(\exists (s,p)\in{\mathbb I}_m\colon \forall \ell\in \{1,\dots,N\}\colon x^\ell\not\in {\mathbb B}_m(s,p)\bigr) 
\,\le\, \#{\mathbb I}_m\cdot (1-2^{-m-4})^N.
\]
As ${\mathbb B}_{m}(s,p)$ was defined in \eqref{eq:def:B} as the intersection of all cubes from $\Omega_m(s,p)$,
finding a point $x^\ell\in {\mathbb B}_m(s,p)$ means that the same point may be found in all cubes in $\Omega_m(s,p).$
We conclude that if $N$ is large enough to ensure that
\[
\exp\Bigl(m2^{m}\log(2^{m+3}d)\Bigr)(1-2^{-m-4})^N< 1,
\]
i.e., $N>m 2^{2m+4} \log(2^{m+3} d)$, 
then the randomly generated  $P=\{x^1,\dots,x^N\}\subset {\mathbb M}_m^d$  
intersects every ${\mathbb B}\in\Omega_m$ with positive probability.
Hence, there exists $P$ with $\# P\le N$ such that $\disp(P)\le2^{-m}$.
This is essentially the result of \cite{UV}.

\subsection{Connection to $k$-restriction problems}

By what we said above, if a point set $P\subset[0,1]^d$ intersects ${\mathbb B}_{m}(s,p)$
for all $(s,p)\in{\mathbb I}_m$, then $\disp(P)\le 2^{-m}.$ The randomized construction in \cite{UV},
which we now want to replace by a deterministic one, was restricted in its choice of points to ${\mathbb M}_m^d$.
Therefore we define
$$
{\mathbb C}_m(s,p)=2^m [{\mathbb B}_m(s,p)\cap {\mathbb M}_m^d],\quad (s,p)\in {\mathbb I}_m.
$$
\begin{definition} Let $m\ge 2$. We say that $T\subset\{1,2,\dots,2^{m}-1\}^d$
satisfies \emph{the condition (S') of the order $m$} if, for every $(s,p)\in{\mathbb I}_m$
it intersects ${\mathbb C}_m(s,p).$
\end{definition}
The rest of \cite{UV} then provides a randomized construction of a small set $T$, which satisfies the condition $(S')$
of order $m$. This task has two things in common with $k$-restriction problems. First, the system $\{{\mathbb C}_m(s,p):(s,p)\in{\mathbb I}_m\}$
is invariant under permutations of $\{1,\dots,d\}$ and second, the number of active coordinates is, for every $(s,p)\in{\mathbb I}_m$,
bounded from above by $A_m:=\min(m2^m,d).$

To build the connection between the condition $(S')$ and the $k$-restriction problems,
we choose the quadruplet of parameters $(b,k,n,m)$, see Definition~\ref{def:krestriction}, 
as $(2^m-1, A_m,d,M)$.
The system ${\mathcal C}$ collects the sets ${\mathbb C}_m(s,p)$ for those $(s,p)\in{\mathbb I}_m$
which have the corresponding active coordinates in $\{1,\dots,A_m\}$. Finally, $M$ is the cardinality of ${\mathcal C}.$

More formally, let $(s,p)\in{\mathbb I}_m$ with $s_j=2^m-1$ for $j>A_m$. 
Then, we set
\[
C_m{(s,p)}={\mathbb C}_{m}(s,p)-1
\]
and define
\begin{align}\label{eq:systemc}
{\mathcal C}&=\Bigl\{C_m{(s,p)}:(s,p)\in{\mathbb I}_{m} \text{ with }s_j=2^m-1\text{ for }j>A_m\Bigr\}.
\end{align}

We observe that a set $T$ satisfies the condition $(S')$ of order $m$ if, and only if, the set
$T-1$ satisfies the $k$-restriction problem with respect to ${\mathcal C}.$

The parameter $M$, which is just the cardinality of ${\mathcal C}$, can be estimated from 
above in a way similar to \eqref{eq:card}, but note that we do not have to choose the subset 
of active indices anymore. 
Each $C_m{(s,p)}\in{\mathcal C}$ is characterized by $s\in\{0,1,\dots,2^{m}-1\}^{A_m}$ and 
$p\in\{1/2^m,\dots,(2^m-1)/2^m\}^{A_m}$.
Therefore, 
\begin{equation}\label{eq:card2}
M\le 2^{2mA_m}.
\end{equation}

The second important parameter of a $k$-restriction problem is the minimal size 
$c:=c(\mathcal{C})$ of each of the restriction sets $C_m(s,p)$, see~\eqref{eq:c}.
A lower bound on $c$ follows directly from Lemma~\ref{lem:discrete} and we obtain 
\begin{equation}\label{eq:c2}
c \,\ge\, {\mathbb P}(z\in {\mathbb B}_m(s,p))\cdot \#{\mathbb M}_m^{A_m}
\,\ge\, 2^{-m-4}(2^m-1)^{A_m}.
\end{equation}
With the choice of parameters as given above, we have $c/b^k\ge 2^{-m-4}$.

\subsection{A first attempt for derandomization}

Using the arguments of the last section, one could use the construction from 
Theorem~\ref{thm:NSS2} directly to solve the corresponding $k$-restriction problem 
with parameters $(2^m-1, A_m,d,2^{2m A_m})$, whenever $d>2^m$. 
This leads to a point set $P\subset\mathbb{M}_m^d$ with $\disp(P)\le2^{-m}$ and
\[
\#P \,\le\, \left\lceil 2^{m+4}\log\bigl(d^{A_m} 2^{2m A_m}\bigr) \right\rceil 
\,=\, \mathcal{O}(m^2 2^{2m} \log(d)).
\]
Note that this bound matches the union bound from Section~\ref{subsec:UV-random}. 
However, the running-time of the algorithm, as given by Theorem~\ref{thm:NSS2}, 
is 
\[
\mathcal{O}\left(2^{m} d^{2A_m} 2^{2mA_m} \mathcal{T}\right)
\,=\,\mathcal{O}\left(2^{m(1+m2^{m+1})} d^{2m^2 2^m} \mathcal{T}\right), 
\]
where $\mathcal{T}$ is the time complexity of the membership oracle, which can be assumed of the order
$\mathcal{O}(A_m)=\mathcal{O}(m2^m)$ in this case.

\medskip

\subsection{Derandomization using splitters}

We now describe how we can improve the construction of an explicit solution to the 
desired $k$-restriction problem with parameters $(b,k,n,m)$ equal to 
$(2^m-1,A_m,d,2^{2m A_m})$
and the set system ${\mathcal C}$ defined by \eqref{eq:systemc}.
We use the approach of~\cite{NSS} to 
obtain solutions of the $k$-restriction problem which are 'small' in size and 
running-time of the corresponding algorithm. 'Small' means here, that the dependence 
on the original problem dimension $d$ is as small as possible. 

In the heart of the constructions are splitters, see Section~\ref{subsec:splitter}.
As already indicated in Section~\ref{subsec:splitter}, we use a $(d,A_m,A_m^2)$-splitter, 
say $A(m,d)$, to map the original $d$-dimensional problem to a $k$-restriction problem 
in dimension $A_m^2$, which can then be solved with cost independent of $d$.

Recall that, by Definition~\ref{def:splitters}, $A(m,d)$ is an $(d,A_m,A_m^2)$-splitter, 
if it is a collection of mappings $a:\{1,\dots,d\}\to\{1,\dots,A_m^2\}$ such that 
for every $S\subset \{1,\dots,d\}$ with $\#S=A_m$, there is an $a\in A(m,d)$,
which is injective on $S$, i.e., $a(S)\subset\{1,\dots,A_m^2\}$ has $A_m$ elements. 

Further, let $T(m)\subset \{1,2,\dots,2^m-1\}^{A_m^2}$ be the solution of the 
$k$-restriction problem with parameters $(2^m-1,A_m,A_m^2,2^{2m A_m})$ 
with respect to the original system of restrictions ${\mathcal C}$, 
see~\eqref{eq:systemc}. 
This means, that
$T(m)=\{\tau_1,\dots,\tau_K\}\subset\{0,1,\dots,2^m-2\}^{A_m^2}$ 
such that for every $S'\subset\{1,\dots,A_m^2\}$ with $\#S'=A_m$ and any 
$C\in\mathcal{C}$ there is $\tau\in T(m)$ with $\tau|_{S'}\in C$.

Now we are in the position to define the solution to the restriction problem with 
parameters $(2^m-1,A_m,d,2^{2m A_m})$ and the system ${\mathcal C}$. Indeed, we define
\[\begin{split}
T^*&=T(m)\circ A(m,d)\\
&:=\{\tau\circ a:\{1,\dots, d\}\to \{0,1,\dots,2^m-2\}:\tau\in T(m), a\in A(m,d)\}
\end{split}\]
to be the set of concatenations of any splitter with any element of the 
solution to the restriction problem. 
Here, we switch between the notion of vectors of length $d$ (resp.~$A_m^2$) 
and mappings from $\{1,\dots,d\}$ (resp.~$\{1,\dots,A_m^2\}$) to $\{0,1,\dots,2^m-2\}$. 
This should not lead to any confusion.

To show that $T^*$ is indeed a solution to our $k$-restriction problem, 
let $S\subset \{1,\dots,d\}$ with $\# S=A_m$. 
Then, there exists $a\in A(m,d)$, such that $S'=a(S)\subset\{1,2,\dots,A_m^2\}$ 
has $A_m$ mutually different elements, i.e., $\#S'=A_m$. 
Now, for every $C\in\mathcal{C}$, there is some $\tau\in T(m)$, 
such that $C\ni\tau|_{S'}=(\tau\circ a)|_S$. 
Hence, $T^*$, which satisfies $\#T^*=\#T(m)\cdot \#A(m,d)$, 
is a solution to the restriction problem with parameters $(2^m-1,A_m,d,2^{2m A_m})$.

We merge all the components together in a form of an algorithm.

\vskip.5cm
\fbox{\begin{minipage}{13cm}
{\bf Algorithm 2}
\begin{enumerate}
\item For $\varepsilon\in(0,\frac{1}{4}]$ and $d\ge2$, choose a positive integer $m$ \\
			with $2^{-m}\le\varepsilon<2^{-m+1}$ and set $A_m:=\min(m2^n,d)$;
\item Generate a $(d,A_m,A_m^2)$-splitter $A(m,d)$ as in Lemma~\ref{lem:splitter};
\item Generate a solution $T(m)$ to the $k$-restriction problem with parameters 
			$(2^m-1,A_m,A_m^2,2^{2m A_m})$ and restrictions $\mathcal{C}$ from~\eqref{eq:systemc} 
			as in Theorem~\ref{thm:NSS2};
\item Set $T^* = T(m)\circ A(m,d) \subset \{0,\dots,2^m-2\}^d$;
\item Increase all the coordinates by one, then divide them by $2^m$;
\item Output the resulting point set $P$.
\end{enumerate}
\end{minipage}}
\vskip.5cm

\begin{theorem}\label{thm:main2}
Let $\varepsilon\in(0,\frac{1}{4}]$ and $d\ge 2$. 
Then, there is an absolute constant $C<\infty$, such that
Algorithm 2 constructs a set $P\subset[0,1]^d$ with $\disp(P)\le \varepsilon$ and 
\[
\#P\le C \left(\frac{1+\log_2(\varepsilon^{-1})}{\varepsilon}\right)^6 \log(d^*),
\]
where $d^*:=\max(d, 2/\eps)$. 

The running-time of Algorithm~2 is $\mathcal{O}_\eps(d^c)$ for some $c<\infty$.
\end{theorem}
\medskip

\begin{proof}
In case that $d<\lfloor 2/\eps\rfloor$, replace $d$ by $\lfloor 2/\eps\rfloor$ 
in all what follows and delete at the end the last $\lfloor 2/\eps\rfloor-d$ 
coordinates of the constructed point set.

By Lemma~\ref{lem:splitter}, there is an explicit, deterministic construction 
of a $(d,A_m,A_m^2)$-splitter $A(m,d)$ of size 
$\mathcal{O}(A_m^4\log(d))=\mathcal{O}(m^4 2^{4m}\log(d))$. 
Furthermore, by Theorem~\ref{thm:NSS2} and since 
$d\ge\lfloor 2/\eps\rfloor> 2^m-1$, there is a deterministic solution $T(m)$
to the $k$-restriction problem with parameters $(2^m-1,A_m,A_m^2,2^{2m A_m})$ 
and restrictions $\mathcal{C}$ from~\eqref{eq:systemc}
of size
\[
\#T(m) \,=\, \mathcal{O}\left(2^m \log(A_m^{2A_m} 2^{2m A_m})\right) 
\,=\, \mathcal{O}\left(m^2 2^{2m}\right).
\]
Here, we also used~\eqref{eq:c2}.
Altogether, we get
\[
\#T^* \,=\, \#T(m)\cdot \#A(m,d) 
\,=\, \mathcal{O}\left(m^6 2^{6m} \log(d)\right),
\]
which implies the result since $m\le 1+\log_2(1/\eps)$.

For the running-time, we need 
$\mathcal{O}_m(d^c)$ 
for the construction of $A(m,d)$, see Lemma~\ref{lem:splitter}, 
and 
\[
\mathcal{O}(2^m(e A_m)^{A_m} 2^{2m A_m}\mathcal{T})
\,=\, \mathcal{O}(2^{4m^2 2^m}\mathcal{T})
\] 
for the construction of $T(m)$, see Theorem~\ref{thm:NSS2}, 
where, again, the cost $\mathcal{T}$ of the membership oracle can be 
assumed to be $\mathcal{O}(A_m)$.
Finally, we need $\mathcal{O}(\#T^*)$ time to build up $T^*$.
The remaining steps are less expensive, which implies the result. \\
\end{proof}

\bigskip


\begin{thebibliography}{99}

\bibitem{AHR15}
C. Aistleitner, A. Hinrichs, and D. Rudolf,
\emph{On the size of the largest empty box amidst a point set},
Discrete Appl. Math. 230 (2017), 146--150.

\bibitem{Alon} N. Alon, \emph{Explicit construction of exponential sized families of $k$-independent sets}, Discr. Math. 58(2) (1986), 191--193.
\bibitem{Alonall} N. Alon, J. Bruck, J. Naor, M. Naor, and R.M. Roth,
\emph{Construction of asymptotically good low-rate error-correcting codes through pseudo-random graphs}, IEEE Trans.  Inf. Theory 38 (2) (1992), 509--516.
\bibitem{AMS} N. Alon, D. Moshkovitz, and S. Safra, \emph{Algorithmic construction of sets for $k$-restrictions}, ACM Trans. Algor. 2 (2006), 153--177.

%
\bibitem{BDDG14}
M. Bachmayr, W. Dahmen, R. DeVore, and L. Grasedyck, 
\emph{Approximation of high-dimensional rank one tensors}, 
Constr. Approx. 39(2) (2014), 385--395.
%

\bibitem{BK} J. Backer and M. Keil, \emph{The mono- and bichromatic empty rectangle and square problems
in all dimensions}, In Proc. 9th Latin American Sympos. on Theor. Informatics, pp. 14-–25, 2010.

\bibitem{CDL86}
B. Chazelle, R. L. Drysdale, and D. T. Lee, \emph{Computing the largest empty rectangle}, 
SIAM J. Comput. 15(1) (1986), 300--315.

\bibitem{DP10} 
J. Dick and F. Pillichshammer,
\emph{Digital nets and sequences}, 
Cambridge University Press, Cambridge, 2010.
%
\bibitem{DP14} 
J. Dick and F. Pillichshammer,
\emph{Discrepancy theory and quasi-Monte Carlo integration}, 
A panorama of discrepancy theory, Lecture Notes in Math. 2107, 
Springer Verlag, 2014.
%
\bibitem{DGW10}
B. Doerr, M. Gnewuch, and M. Wahlstr\"om, 
\emph{Algorithmic construction of low-discrepancy point sets via
dependent randomized rounding}, J. Complexity 26 (2010), 490--507.

\bibitem{DGW14} C. Doerr, M. Gnewuch, and M. Wahlstr\"om,
\emph{Calculation of discrepancy measures and applications}, A panorama of discrepancy theory,
Lecture Notes in Math. 2107,
Springer Verlag, 2014.

\bibitem{DT97}
M. Drmota and R. F. Tichy, \emph{Sequences, discrepancies and applications},
Lecture Notes in Math. 1651, Springer Verlag, 1997.

\bibitem{DJ13}
A. Dumitrescu and M. Jiang, \emph{Maximal empty boxes amidst random points}, 
Combin. Probab. Comput. 22(4) (2013), 477--498.

\bibitem{DJ13b}
A. Dumitrescu and M. Jiang, \emph{On the largest empty axis-parallel box amidst n points}, 
Algorithmica 66(2) (2013), 225--248.

\bibitem{DJ16}
A. Dumitrescu and M. Jiang, \emph{Perfect vector sets, properly overlapping partitions, and largest empty box}, 
arXiv:1608.06874v1, 2016.

\bibitem{DJ18}
A. Dumitrescu and M. Jiang, \emph{On the number of maximum empty boxes amidst n points}, 
Discrete Comput. Geom.  59(3) (2018), 742--756.

\bibitem{G10} M. Gnewuch, \emph{Entropy, randomization, derandomization, and discrepancy},
Monte Carlo and quasi-Monte Carlo methods 2010, Springer Proc. Math. Stat. 23, Springer Verlag, 2012.

\bibitem{GSW09}
M. Gnewuch, A. Srivastav, and C. Winzen, 
\emph{Finding optimal volume subintervals with k points and calculating the star
discrepancy are NP-hard problems}, J. Complexity 25 (2009), 115--127.

\bibitem{HPUV} A. Hinrichs, J. Prochno, M. Ullrich, and J. Vyb\'iral, \emph{The minimal $k$-dispersion of point sets in high-dimensions}, to appear in: J. Complexity.

\bibitem{KS} D. J. Kleitman and J. H. Spencer, \emph{Families of $k$-independent sets}, Discrete Math. 6 (1973), 255--262.

\bibitem{Kr17}
D. Krieg,
\emph{On the dispersion of sparse grids},
J. Complexity 45 (2018), 115--119.

\bibitem{KR18}
D. Krieg and D. Rudolf. 
\emph{Recovery algorithms for high-dimensional rank one tensors},
J. Approx. Theory 237 (2019), 17--29.

\bibitem{LPS} A. Lubotzky, R. Phillips,  and P. Sarnak, \emph{Explicit expanders and the Ramanujan conjectures},  in Proc.~18th ACM Symp. Theory of Comput. (1986), 240--246;  
	See also: \emph{Ramanujan  graphs}, Combinarorica 8 (1988), 261--277. 
	
\bibitem{NLH84} A. Naamad, D. Lee, and W.-L. Hsu. \emph{On the maximum empty rectangle problem},
Discrete Appl. Math. 8(3) (1984), 267--277.

\bibitem{NN} J. Naor and M. Naor, \emph{Small-bias probability spaces: Efficient constructions and applications}, SIAM J. Comp. 22(4) (1993), 838--856.

\bibitem{NSS} M. Naor, L. J. Schulman, and A. Srinivasan, 
\emph{Splitters and near-optimal derandomization}, Foundations of Computer Science, IEEE Proceedings of the 36th Annual Symposium, 182--191, 1995.

\bibitem{Niederreiter83}       
H. Niederreiter,        
\emph{A quasi-{M}onte {C}arlo method for the approximate computation of the extreme values of a function}, Studies in Pure Mathematics, pp.~523--529, Birkh\"auser, Basel, 1983.
%
\bibitem{Niederreiter92}       
H. Niederreiter,        
\emph{Random number generation and quasi-Monte Carlo methods},       
Society for Industrial and Applied Mathematics, Philadelphia, 1992.
%
\bibitem{No15}       
E. Novak,        
\emph{Some results on the complexity of numerical integration},
In: R. Cools, D. Nuyens (eds) Monte Carlo and quasi-Monte Carlo methods,
Proceedings in Mathematics \& Statistics, vol 163, Springer Verlag, 2016.
%
\bibitem{NR15}
E. Novak and D. Rudolf, 
\emph{Tractability of the approximation of high-dimensional rank one tensors}, 
Constr. Approx. 43(1) (2016), 1--13.
%
\bibitem{NW10}      
E. Novak and H. Wo\'zniakowski,             
\emph{Tractability of Multivariate Problems},      
Volume II: Standard Information for Functionals,      
European Math. Soc. Publ. House, Z\"urich,              
2010.   

\bibitem{RT96}
G. Rote and R. F. Tichy, 
\emph{Quasi-Monte Carlo methods and the dispersion of point sequences}, 
Math. Comput. Modelling 23(8-9) (1996), 9--23.

\bibitem{Ru17}
D. Rudolf, 
\emph{An upper bound of the minimal dispersion via delta covers},
Contemporary Computational Mathematics -- a Celebration of the 80th 
Birthday of Ian Sloan. Springer-Verlag, 2018.

\bibitem{SB} G. Seroussi and N.H. Bshouty, \emph{Vector sets for exhaustive testing of logic circuits}, IEEE Trans. Inform. Theory 34 (1988), 513--522.
\bibitem{Sosno} J. Sosnovec, \emph{A note on the minimal dispersion of point sets in the unit cube}, Eur. J. Combin. 69 (2018), 255--259.

\bibitem{Te17b}
V.N. Temlyakov, 
\emph{The Marcinkiewicz-type discretization theorems,}
Constr. Approx. 48 (2018), 337--369.

\bibitem{Te17c}
V.N. Temlyakov,
\emph{Universal discretization,}
J. Complexity 47 (2018), 97--109.

\bibitem{MU18}
M. Ullrich,
\emph{A lower bound for the dispersion on the torus,}
Math. Comput. Simulation 143 (2018), 186--190.

\bibitem{UV} M. Ullrich and J. Vyb\'iral, \emph{An upper bound on the minimal dispersion}, J. Complexity 45 (2018), 120--126.

\bibitem{YLV00}
S. Yakowitz, P. L'Ecuyer, and F. V{\'a}zquez-Abad,
\emph{Global stochastic optimization with low-dispersion point sets}, 
Oper. Res. 48(6) (2000), 939--950.

\end{thebibliography}
\end{document}